\documentclass[aps,pra,superscriptaddress,twocolumn,footinbib]{revtex4-1}

\usepackage{amsthm,amsmath,amssymb,amsbsy}

\usepackage{color}
\definecolor{darkred}{rgb}{0.8,0.1,0.1}
\definecolor{lightblue}{rgb}{0.1,0.1,0.8}

\usepackage{graphicx,epic,eepic,epsfig,verbatim,color}
 \usepackage{tikz}

\usepackage[T1]{fontenc}
\usepackage[utf8]{inputenc}

 \usepackage[breaklinks]{hyperref}

\newtheorem{proposition}{Proposition}
\newtheorem{lemma}{Lemma}

\newtheorem{theorem}{Theorem}

\usepackage{cleveref}
\usepackage{graphicx}
\usepackage{xcolor}

\definecolor{darkblue}{RGB}{0,76,156}
\definecolor{darkkblue}{RGB}{0,0,153}
\definecolor{blue2}{RGB}{102,178,255}

\def\endenv{\ifmmode\;\else{\unskip\nobreak\hfil
\penalty50\hskip1em\null\nobreak\hfil\;
\parfillskip=0pt\finalhyphendemerits=0\endgraf}\fi}
\newenvironment{remark}{\noindent \textbf{{Remark~}}}{}

\mathchardef\ordinarycolon\mathcode`\:
\mathcode`\:=\string"8000
\def\vcentcolon{\mathrel{\mathop\ordinarycolon}}
\begingroup \catcode`\:=\active
  \lowercase{\endgroup
  \let :\vcentcolon
  }

\makeatletter
\def\resetMathstrut@{%
  \setbox\z@\hbox{%
    \mathchardef\@tempa\mathcode`\[\relax
    \def\@tempb##1"##2##3{\the\textfont"##3\char"}%
    \expandafter\@tempb\meaning\@tempa \relax
  }%
  \ht\Mathstrutbox@\ht\z@ \dp\Mathstrutbox@\dp\z@}
\makeatother
\begingroup
  \catcode`(\active \xdef({\left\string(}
  \catcode`)\active \xdef){\right\string)}
\endgroup
\mathcode`(="8000 \mathcode`)="8000

\newcommand{\nc}{\newcommand}
\nc{\rnc}{\renewcommand}
\nc{\beg}{\begin{equation}}
\nc{\eeq}{{\end{equation}}}
\nc{\beqa}{\begin{eqnarray}}
\nc{\eeqa}{\end{eqnarray}}
\nc{\lbar}[1]{\overline{#1}}
\nc{\ketbra}[2]{|#1\rangle\!\langle#2|}

\nc{\avg}[1]{\langle#1\rangle}
\nc{\Rank}{\operatorname{Rank}}
\nc{\smfrac}[2]{\mbox{$\frac{#1}{#2}$}}
\nc{\tr}{\operatorname{Tr}}
\nc{\ox}{\otimes}
\nc{\dg}{\dagger}
\nc{\dn}{\downarrow}
\nc{\cA}{{\cal A}}
\nc{\cB}{{\cal B}}
\nc{\cC}{{\cal C}}
\nc{\cD}{{\cal D}}
\nc{\cE}{{\cal E}}
\nc{\cF}{{\cal F}}
\nc{\cG}{{\cal G}}
\nc{\cH}{{\cal H}}
\nc{\cI}{{\cal I}}
\nc{\cJ}{{\cal J}}
\nc{\cK}{{\cal K}}
\nc{\cL}{{\cal L}}
\nc{\cM}{{\cal M}}
\nc{\cN}{{\cal N}}
\nc{\cO}{{\cal O}}
\nc{\cP}{{\cal P}}
\nc{\cQ}{{\cal Q}}
\nc{\cR}{{\cal R}}
\nc{\cS}{{\cal S}}
\nc{\cT}{{\cal T}}
\nc{\cV}{{\cal V}}
\nc{\cU}{{\cal U}}
\nc{\cX}{{\cal X}}
\nc{\cY}{{\cal Y}}
\nc{\cZ}{{\cal Z}}
\nc{\cW}{{\cal W}}
\nc{\csupp}{{\operatorname{csupp}}}
\nc{\qsupp}{{\operatorname{qsupp}}}
\nc{\var}{{\operatorname{var}}}
\nc{\rar}{\rightarrow}
\nc{\lrar}{\longrightarrow}
\nc{\polylog}{{\operatorname{polylog}}}
\nc{\1}{{\mathds{1}}}
\nc{\wt}{{\operatorname{wt}}}
\nc{\av}[1]{{\left\langle {#1} \right\rangle}}
\nc{\supp}{{\operatorname{supp}}}

\def\ve{\varepsilon}

\def\x{\xi}

\def\D{\Delta}
\def\T{\Theta}

\def\O{\Omega}

\nc{\RR}{{{\mathbb R}}}
\nc{\CC}{{{\mathbb C}}}
\nc{\FF}{{{\mathbb F}}}
\nc{\NN}{{{\mathbb N}}}
\nc{\ZZ}{{{\mathbb Z}}}
\nc{\PP}{{{\mathbb P}}}
\nc{\QQ}{{{\mathbb Q}}}
\nc{\UU}{{{\mathbb U}}}
\nc{\EE}{{{\mathbb E}}}

\nc{\CHSH}{{\operatorname{CHSH}}}

\nc{\be}{\begin{equation}}
\nc{\ee}{{\end{equation}}}
\nc{\bea}{\begin{eqnarray}}
\nc{\eea}{\end{eqnarray}}
\nc{\Hom}[2]{\mbox{Hom}(\CC^{#1},\CC^{#2})}
\nc{\rU}{\mbox{U}}

\nc{\ob}[1]{#1}

\nc{\SEP}{{\text{SEP}}}
\nc{\NS}{{\text{NS}}}
\nc{\LOCC}{{\text{LOCC}}}
\nc{\PPT}{{\text{PPT}}}
\nc{\EXT}{{\text{EXT}}}
\nc{\Sym}{{\operatorname{Sym}}}

\nc{\ERLO}{{E_{\text{r,LO}}}}
\nc{\ERLOCC}{{E_{\text{r,LOCC}}}}
\nc{\ERPPT}{{E_{\text{r,PPT}}}}
\nc{\ERLOCCinfty}{{E^{\infty}_{\text{r,LOCC}}}}
\nc{\Aram}{{\operatorname{\sf A}}}

\usepackage{tikz}
\usetikzlibrary{arrows}

\makeatletter
\def\grd@save@target#1{%
  \def\grd@target{#1}}
\def\grd@save@start#1{%
  \def\grd@start{#1}}
\tikzset{
  grid with coordinates/.style={
    to path={%
      \pgfextra{%
        \edef\grd@@target{(\tikztotarget)}%
        \tikz@scan@one@point\grd@save@target\grd@@target\relax
        \edef\grd@@start{(\tikztostart)}%
        \tikz@scan@one@point\grd@save@start\grd@@start\relax
        \draw[minor help lines,magenta] (\tikztostart) grid (\tikztotarget);
        \draw[major help lines] (\tikztostart) grid (\tikztotarget);
        \grd@start
        \pgfmathsetmacro{\grd@xa}{\the\pgf@x/1cm}
        \pgfmathsetmacro{\grd@ya}{\the\pgf@y/1cm}
        \grd@target
        \pgfmathsetmacro{\grd@xb}{\the\pgf@x/1cm}
        \pgfmathsetmacro{\grd@yb}{\the\pgf@y/1cm}
        \pgfmathsetmacro{\grd@xc}{\grd@xa + \pgfkeysvalueof{/tikz/grid with coordinates/major step}}
        \pgfmathsetmacro{\grd@yc}{\grd@ya + \pgfkeysvalueof{/tikz/grid with coordinates/major step}}
        \foreach \x in {\grd@xa,\grd@xc,...,\grd@xb}
        \node[anchor=north] at (\x,\grd@ya) {\pgfmathprintnumber{\x}};
        \foreach \y in {\grd@ya,\grd@yc,...,\grd@yb}
        \node[anchor=east] at (\grd@xa,\y) {\pgfmathprintnumber{\y}};
      }
    }
  },
  minor help lines/.style={
    help lines,
    step=\pgfkeysvalueof{/tikz/grid with coordinates/minor step}
  },
  major help lines/.style={
    help lines,
    line width=\pgfkeysvalueof{/tikz/grid with coordinates/major line width},
    step=\pgfkeysvalueof{/tikz/grid with coordinates/major step}
  },
  grid with coordinates/.cd,
  minor step/.initial=.2,
  major step/.initial=1,
  major line width/.initial=2pt,
}
\makeatother

\tikzset{
  treenode/.style = {align=center, inner sep=0pt, text centered,
    font=\sffamily},
  arn_n/.style = {treenode, circle, white, font=\sffamily\bfseries, draw=black,
    fill=black, text width=1.5em},
  arn_r/.style = {treenode, circle, red, draw=red, 
    text width=1.5em, very thick},
  arn_x/.style = {treenode, rectangle, draw=black,
    minimum width=0.5em, minimum height=0.5em}
}

\usepackage{pifont}
\usepackage{graphicx}
\usepackage{bm,bbm}%
\usepackage{braket}
\usepackage{enumerate}
\usepackage{subcaption,caption}
\usepackage{booktabs,multirow}
\usepackage{mathtools,bbding}
\usepackage[percent]{overpic}
\usepackage{microtype}

\usepackage{newpxtext,newpxmath}

\usepackage[most,breakable]{tcolorbox}

\AtBeginDocument{
\heavyrulewidth=.08em
\lightrulewidth=.05em
\cmidrulewidth=.03em
\belowrulesep=.65ex
\belowbottomsep=0pt
\aboverulesep=.4ex
\abovetopsep=0pt
\cmidrulesep=\doublerulesep
\cmidrulekern=.5em
\defaultaddspace=.5em
}

\usepackage{etoolbox}
\hypersetup{linktoc=all}

\DeclareMathOperator{\Tr}{Tr}

\DeclareMathOperator{\conv}{conv}

\newcommand{\norm}[2]{\left\lVert#1\right\rVert_{\,#2}}
\newcommand{\proj}[1]{\ket{#1}\!\bra{#1}}

\newcommand{\lnorm}[2]{\left\lVert#1\right\rVert_{\ell_{#2}}}

\renewcommand{\T}{\mathcal{T}}
\newcommand{\B}{\mathcal{B}}

\newcommand{\C}{\mathcal{C}}

\newcommand{\V}{\mathcal{V}}

\newcommand{\M}{\mathcal{M}}
\newcommand{\N}{\mathcal{N}}
\newcommand{\I}{\mathcal{I}}
\newcommand{\J}{\mathcal{J}}
\newcommand{\Q}{\mathcal{Q}}
\newcommand{\Qm}{{\mathcal{Q}_{m}}}

\let\T\relax
\newcommand{\T}[1]{T^{({#1})}_\I}
\newcommand{\TJ}[1]{T^{({#1})}_\J}

\newcommand{\mleq}{\leq}
\newcommand{\mgeq}{\geq}

\nc{\MIO}{{\text{\rm MIO}}}
\nc{\DIO}{{\text{\rm DIO}}}
\nc{\SIO}{{\text{\rm SIO}}}
\nc{\IO}{{\text{\rm IO}}}

\renewcommand{\*}{\textup{*}}
\newcommand{\<}{\left\langle}
\renewcommand{\>}{\right\rangle}

\renewcommand{\bar}{\;\rule{0pt}{9.5pt}\right|\;}

\newcommand{\lset}{\left\{\left.}
\newcommand{\rset}{\right\}}

\DeclareMathOperator*{\argmin}{arg\,min}

\newcommand{\DD}{\mathbb{D}}

\newcommand{\cbraket}[1]{\left|\braket{#1}\right|}

\newcommand{\id}{\mathbbm{1}}

\renewcommand{\O}{\mathcal{O}}

\newcommand{\mnorm}[1]{\norm{#1}{[m]}}

\let\oldproofname\proofname
\renewcommand{\proofname}{\rm\bf{\oldproofname}}

\makeatletter
\renewenvironment{proof}[1][\proofname]{%
  \vspace{-\topsep}
  \pushQED{\qed}
  \normalfont
  \topsep6\p@\@plus6\p@\relax
  \trivlist\item[\hskip\labelsep\bfseries#1\@addpunct{.}]\ignorespaces}{\popQED\endtrivlist\@endpefalse}
\makeatother

\newcommand{\notts}{\affiliation{School of Mathematical Sciences and Centre for the Mathematics and Theoretical Physics of Quantum Non-Equilibrium Systems, University of Nottingham, University Park, Nottingham NG7 2RD, United Kingdom}}

\begin{document}

\title{One-shot coherence distillation}


\author{Bartosz Regula}
\email{bartosz.regula@gmail.com}
\notts

\author{Kun Fang}
\email{kun.fang-1@student.uts.edu.au}
\affiliation{Centre for Quantum Software and Information, School of Software, Faculty of Engineering and Information Technology, University of Technology Sydney, NSW 2007, Australia}

\author{Xin Wang}
\email{xin.wang-8@student.uts.edu.au}
\affiliation{Centre for Quantum Software and Information, School of Software, Faculty of Engineering and Information Technology, University of Technology Sydney, NSW 2007, Australia}

\author{Gerardo Adesso}
\email{gerardo.adesso@nottingham.ac.uk}
\notts


\begin{abstract}
We characterize the distillation of quantum coherence in the one-shot setting, that is, the conversion of general quantum states into maximally coherent states under different classes of quantum operations. We show that the maximally incoherent operations (MIO) and the dephasing-covariant incoherent operations (DIO) have the same power in the task of one-shot coherence distillation. We establish that the one-shot distillable coherence under MIO and DIO is efficiently computable with a semidefinite program, which we show to correspond to a quantum hypothesis testing problem. Further, we introduce a family of coherence monotones generalizing the robustness of coherence as well as the modified trace distance of coherence, and show that they admit an operational interpretation in characterizing the fidelity of distillation under different classes of operations. By providing an explicit formula for these quantities for pure states, we show that the one-shot distillable coherence under MIO, DIO, strictly incoherent operations (SIO), and incoherent operations (IO) is equal for all pure states.
\end{abstract}

\maketitle

\section{Introduction}

The phenomenon of {\it quantum coherence}, emerging from the fundamental property of quantum superposition, has found use in a variety of physical tasks in quantum cryptography, quantum information processing, thermodynamics, metrology, and even quantum biology \cite{streltsov_2017}. The recent years have seen the development of the resource-theoretic framework of quantum coherence, establishing precise physical and mathematical laws governing the creation, manipulation, and conversion of coherence \cite{aberg_2006,gour_2008,levi_2014,baumgratz_2014}.
The archetypal example of a resource theory is quantum entanglement, and although coherence and entanglement share a large number of similarities which allowed for many parallels and interrelations between the two resource theories to be established \cite{baumgratz_2014,streltsov_2015,sperling_2015,vogel_2014,winter_2016,killoran_2016,streltsov_2016,rana_2017,chitambar_2016-2,zhu_2018-1,zhu_2017-1,regula_2018-2}, the two are also very different in some aspects and can require different approaches. One particular difference is the lack of a single, physically-motivated choice of free operations which best describe the allowed state manipulations in the resource theory of quantum coherence, unlike the standard choice of local operations and classical communication (LOCC) for entanglement \cite{horodecki_2009}. It thus becomes necessary to characterize the operational properties and applications of quantum coherence under several different sets of such operations \cite{winter_2016,chitambar_2016,marvian_2016,vicente_2017,streltsov_2017}, the most common ones being incoherent operations (IO) \cite{baumgratz_2014}, strictly incoherent operations (SIO) \cite{winter_2016}, dephasing-covariant incoherent operations (DIO) \cite{chitambar_2016,marvian_2016}, and maximally incoherent operations (MIO) \cite{aberg_2006}.

One of the most significant aspects of a resource theory are the rules governing state transformations under the free operations. In particular, the problem of \emph{distillation} asks: given a canonical unit of coherence represented by the maximally coherent state $\ket\Psi$, what is the best rate at which we can convert copies of a state $\rho$ into copies of $\ket\Psi$ under a chosen set of free operations? The standard approach to this problem in quantum information theory, both in the resource theories of entanglement \cite{bennett_1996-1,bennett_1996-3,rains_1999} and coherence \cite{winter_2016,chitambar_2016-3}, is to consider the asymptotic limit --- that is, assume that we have access to an unbounded number of independent and identically distributed (i.i.d.) copies of a quantum system. In a realistic setting, however, the resources are finite and the number of i.i.d. prepared states is necessarily limited. More importantly, it is very difficult to perform coherent state manipulations over large numbers of systems. Therefore, it becomes crucial to be able to characterize how well we can distill maximally coherent states from a finite number of copies of the prepared states. The study of such non-asymptotic scenarios has garnered great interest in quantum information theory \cite{wang_2012,renes_2011,tomamichel_2013,berta_2011,leung_2015,datta_2013,wang_2017,fang_2017}, including work in the one-shot theory of entanglement distillation \cite{brandao_2011,buscemi_2010-1,buscemi_2013}. More recently, one-shot results in the resource theory of coherence \cite{zhao_2018,bu_2017} and more general quantum resource theories \cite{gour_2017,gour_2017-1} were obtained. 

In this Letter, we develop the framework for non-asymptotic coherence distillation, in which one has access only to a single copy of a quantum system and allows for a finite accuracy, reflecting the realistic restrictions on state transformations. In particular, we establish an exact expression for the one-shot distillable coherence under MIO and DIO, which can be efficiently computed as a semidefinite program (SDP). Interestingly, we show that the two quantities are in fact the same, demonstrating that MIO and DIO have the same power in the task of coherence distillation, and together with recent results in coherence dilution \cite{zhao_2018,chitambar_2017-1} shedding light on the asymptotic reversibility of state transformations under DIO. Further, we generalize two fundamental quantifiers of coherence, the robustness of coherence \cite{napoli_2016} and the modified trace distance of coherence \cite{yu_2016}, establishing a family of measures of coherence which we show to have an operational application in characterizing the maximal fidelity of distillation under different sets of operations. We derive exact expressions for these measures for all pure states, leading to a complete characterization of pure-state one-shot distillation of coherence and showing that all the considered sets of operations --- IO, SIO, DIO, MIO --- have exactly the same power in such a task. We discuss our methods and results below, and defer more technical derivations to the Supplemental Material \footnote{See the Supplemental Material below.}.

\section{A family of coherence monotones}
Consider a fixed orthonormal basis $\{\ket{i}\}$ in a $d$-dimensional Hilbert space ($d < \infty$). We will use $\DD$ to denote the set of all density matrices in this space, and for a pure state $\ket\psi$ we will write $\psi\coloneqq\proj\psi$. Let $\Delta$ denote the diagonal map (fully dephasing channel) in the basis $\{\ket{i}\}$. We will denote by $\I$ the set of density matrices which are diagonal in this basis, i.e. $\rho \in \DD$ such that $\rho = \Delta(\rho)$, and by $\I\*\*$ the cone of diagonal positive semidefinite matrices which are not necessarily normalized.

The resource theory of coherence consists of the following ingredients \cite{baumgratz_2014}: the set of free \emph{incoherent states}, represented by $\I$, and the free operations, that is, a set of quantum operations which do not generate coherence. The largest possible set of such free operations are the \emph{maximally incoherent operations (MIO)} \cite{aberg_2006}, which are given by quantum channels $\cE$ such that $\cE(\rho) \in \I$ for every $\rho \in \I$. The \emph{incoherent operations (IO)} \cite{baumgratz_2014} are those for which there exists a Kraus decomposition into incoherent Kraus operators, that is, $\{K_\ell\}$ such that
$ K_\ell\rho K_\ell^{\dagger}\in\I\*\*$ for all $\ell$ and all $\rho\in \I$. The \emph{strictly incoherent operations (SIO)} \cite{winter_2016} are operations for which both $\{K_\ell\}$ and $\{K^\dagger_\ell\}$ are sets of incoherent operators. Finally, the \emph{dephasing-covariant incoherent operations (DIO)} are maps $\cE$ such that $[\Delta, \cE] = 0$.
The following strict inclusions hold: $\MIO \supsetneq \IO \supsetneq \SIO$, $\MIO \supsetneq \DIO \supsetneq \SIO$ \cite{chitambar_2016-1,marvian_2016}.

Throughout the development of the resource theory of coherence, many different quantifiers of this resource have been defined \cite{baumgratz_2014,streltsov_2015,rana_2016,napoli_2016,yu_2016,streltsov_2017}. A particular example is the \textit{trace distance of coherence}, given by \cite{baumgratz_2014,rana_2016}
\begin{equation}\begin{aligned}
  T_\I(\rho) = \min \lset \norm{\rho - \sigma}{1} \bar \sigma \in \I \rset.
\end{aligned}\end{equation}
Although the trace distance is a fundamental quantity in quantum information theory \cite{nielsen_2011,wilde_2017}, the trace distance of coherence was found to violate the property of strong monotonicity under incoherent operations \cite{yu_2016}, which is considered as one of the requirements that a valid measure of coherence should satisfy \cite{baumgratz_2014,streltsov_2017}. Therefore, an alternative measure called the \textit{modified trace distance of coherence} satisfying strong monotonicity under IO was proposed \cite{yu_2016}:
\begin{equation}\begin{aligned}\label{eq:mod_trace}
  T'_\I (\rho) &= \min \lset \norm{\rho - \lambda \sigma}{1} \bar \sigma \in \I,\; \lambda \geq 0 \rset\\
  &= \min \lset \norm{\rho -  X}{1} \bar X \in \I\*\* \rset.\\
\end{aligned}\end{equation}
Noting that strong duality holds \cite{regula_2018,johnston_2017-1}, we can consider the Lagrange dual of the above expression to obtain a characterization of the modified trace distance of coherence as
\begin{equation}\begin{aligned}
   T'_\I (\rho) &= \max \lset \< \rho, W \> \bar -\mathbbm{1} \mleq W \mleq \mathbbm{1},\; \Delta(W) \mleq 0 \rset
\end{aligned}\end{equation}
with the Hilbert-Schmidt inner product $\<X,Y\> = \Tr(XY)$ for Hermitian operators. We then extend the above to a family of quantifiers given by SDPs of the form
\begin{equation}\begin{aligned}\label{eq:family_Tm}
 \!\!\! \T{m} (\rho) = \max \big\{ \< \rho, W \> \;\big|\;  &-\mathbbm{1} \mleq W \mleq m\mathbbm{1},\;
  \Delta(W) \mleq 0 \,\big\},\!\!
\end{aligned}\end{equation}
similarly to the approach of Brand\~ao \cite{brandao_2005} for entanglement measures. Here, we will take $m \in \NN$, although it can be treated as a continuous parameter in general. The fact that each such measure is a valid coherence monotone can be shown by expressing $\T{m}$ as a convex gauge function \cite{regula_2018}, and we formalize it as follows.
\begin{proposition}\label{prop:Tm_monotonic}
For each $m\geq 1$, $\T{m}$ is a faithful and convex coherence measure satisfying strong monotonicity under MIO.
 \end{proposition}

Note that $m=1$ gives $T'_\I(\rho)$. For $m=d-1$, one can notice that the constraint $W \mleq (d-1)\mathbbm{1}$ is redundant: the other constraints ensure that the smallest eigenvalue of $W$ is at least $-1$ and that the trace of $W$ is at most $0$, together implying that there cannot exist an eigenvalue of $W$ which is larger than $d-1$. Therefore, we get
\begin{equation}\begin{aligned}
  \T{d-1}(\rho) &= \max \lset \< \rho, W\> \bar\! -\mathbbm{1} \mleq W,\; \Delta(W) \mleq 0 \ \rset \\
  &= R_\I(\rho)
\end{aligned}\end{equation}
where $R_\I(\rho)$ is the robustness of coherence \cite{napoli_2016,piani_2016}. This shows that $\T{m}$ can be thought of as a family of measures interpolating between the modified trace distance and the robustness of coherence. Notice that we clearly also have $\T{m}(\rho) = R_\I(\rho)$ for any $m > d-1$. In general, we have that $0 \leq \T{m}(\rho) \leq m \;\forall \rho \in \DD$. It is straightforward to see by strong Lagrange duality that the family $\T{m}$ satisfies
\begin{equation}\begin{aligned}\label{eq:modified-trace-primal}
  \T{m}(\rho) = \min_{X \in \I\*\*} \;m \Tr\left(\rho-X\right)_+ + \Tr\left(\rho-X\right)_-\,,
\end{aligned}\end{equation}
where $(\rho-X)_\pm$ denote the positive and negative parts of the Hermitian operator $\rho-X$.

To characterize the values of the family of quantifiers $\T{m}$ on pure states, we consider the following quantity:
\begin{equation}\begin{aligned}\label{eq:norm_primal}
  \mnorm{\ket\psi} = \min_{\ket{\psi} = \ket{x} + \ket{y}} \lnorm{\ket{x}}{1} + \sqrt{m} \lnorm{\ket{y}}{2}\,,
\end{aligned}\end{equation}
which we (for reasons that will become clear later) call the $m$\textit{-distillation norm}. 
We can then obtain the following result.
\begin{theorem}\label{thm:maj_norm_trace_distance}
For any pure state $\ket\psi$ and any $m \geq 1$,
\begin{equation}\begin{aligned}
  \T{m-1}(\psi) = \mnorm{\ket\psi}^2 - 1.
\end{aligned}\end{equation}
\end{theorem}
The proof of this Theorem relies on the fact that each $\T{m-1}$ can be viewed as a robustness measure $R_\Qm$ defined with respect to the set $\Qm \coloneqq \I \cup \frac{1}{m}\DD$. Each such quantity was shown in \cite{regula_2018} to reduce on pure states to a corresponding norm defined at the level of the underlying Hilbert space --- in this case, it is precisely the $m$-distillation norm $\mnorm{\ket\psi}$.

A property of the $m$-distillation norm which will be crucial in the characterization of coherence distillation is that it can, in fact, be computed exactly. In particular, the following holds.
\begin{theorem}\label{thm:maj_norm_formula}
For a pure state $\ket\psi$, let $\psi^\downarrow_{1:k}$ denote the vector consisting of the $k$ largest (by magnitude) coefficients of $\ket\psi$, and analogously let $\psi^\downarrow_{k+1:d}$ denote the vector of the $d-k$ smallest coefficients of $\ket\psi$, with $\psi^\downarrow_{1:0}$ being the zero vector. Then, for any pure state $\ket\psi$ and any integer $m \in \{1,\ldots,d\}$ we have
\begin{equation*}\begin{aligned}
  \mnorm{\ket\psi} = \lnorm{\psi^\downarrow_{1:m-k^\star}}{1} + \sqrt{k^\star} \lnorm{\psi^\downarrow_{m-k^\star+1:d}}{2},
\end{aligned}\end{equation*}
where $k^\star = \displaystyle\argmin_{1 \leq k \leq m} \mbox{$\frac{\lnorm{\psi^\downarrow_{m-k+1:d}}{2}^2}{k}$}.$
\end{theorem}
Theorem \ref{thm:maj_norm_formula} generalizes the recent result of Johnston et al. \cite{johnston_2017-1}, where an explicit formula for the modified trace distance $\T{1}$ was obtained for all pure states. Notice in particular that, if all coefficients of $\ket\psi$ satisfy $|\psi_i| \leq \frac{1}{\sqrt{m}}$, then $k^\star = m$ is optimal and we have $\mnorm{\ket\psi} = \sqrt{m}$, which means that $\T{m-1}(\psi)$ reaches its maximum value $m-1$. As a consequence, $\T{m-1}$ does not, in general, admit a unique maximizer in the form of the maximally coherent state, which in \cite{johnston_2017-1} was considered as a possible indication that this quantity is not suitable as a coherence quantifier. In the following, however, we will instead demonstrate its operational usefulness in the characterization of the fidelity of one-shot coherence distillation.


\section{Distillation of coherence}
We will denote by $\Psi_m = \proj{\Psi_m}$ the $m$-dimensional maximally coherent state $\ket{\Psi_m} = \sum_{i=1}^{m} \frac{1}{\sqrt{m}} \ket{i}$ in the reference basis. The {\em distillable coherence} $C^\infty_{d,\IO}(\rho)$ is the asymptotic rate at which $\Psi_2$ can be obtained per copy of a given state $\rho$ via incoherent operations.
Winter and Yang~\cite{winter_2016} showed
that the distillable coherence of an
arbitrary mixed state coincides with the \emph{relative
entropy of coherence} $C_r(\rho) \coloneq \min_{\sigma\in \I} D(\rho\|\sigma)$ introduced in~\cite{aberg_2006}, where the quantum
relative entropy is given as $D(\rho\|\sigma)\coloneq\tr\rho(\log\rho-\log\sigma)$ with the logarithm taken in base 2.
For any state $\rho$, the distillable coherence is then given by $C^\infty_{d,\IO}(\rho) = C_r(\rho) = S(\D(\rho)) - S(\rho)$.

We now consider the non-asymptotic setting. For any quantum state $\rho$, the fidelity of coherence distillation under the class of operations $\O$ is defined by
  \begin{align}
       F_{\O}(\rho,m) \coloneq \max_{\Lambda \in \O} \< \Lambda(\rho), \Psi_m \>.
     \end{align}
The one-shot $\ve$-error distillable coherence is then defined as the maximum over all distillation rates achievable under the given class of operations with an error tolerance of $\ve$, that is,
\begin{equation}
C_{d,\O}^{(1),\ve}(\rho) := \log \max \lset m \in \mathbb N \bar F_\O(\rho,m) \geq 1- \ve \rset.
\label{one shot distillable coherence}
\end{equation}
As a consequence, the asymptotic distillable coherence can be given as
\begin{align}
  C^\infty_{d,\O}(\rho) = \lim_{\ve \to 0} \lim_{n \to \infty} \frac{1}{n} C_{d,\O}^{(1),\ve}(\rho^{\ox n}).
\end{align}

One of the main results of this work is that the one-shot distillable coherence can be computed exactly, as characterized in the following result.
\begin{theorem}
\label{thm:distillable_SDP}
  If $\O \in \{ \MIO, \DIO \}$, then for any state $\rho \in \DD$, the fidelity of coherence distillation and one-shot $\ve$-error distillable coherence can both be written as the following semidefinite programs:
  \begin{equation}\begin{aligned}
    F_{\O}(\rho,m)  =  \max \left\{ \<G, \rho\> \;\left|\; 0 \leq G \leq \id,\ \Delta(G) = \frac1m \id \right.\right\},
     \end{aligned}\end{equation}\vspace{-\baselineskip}\begin{equation*}\begin{aligned}\label{eq:dist_coherence}
  C_{d,\rm \O}^{(1),\ve}(\rho)  = \log \bigg\lfloor \max \bigg\{& m \;\bigg|\; \< G, \rho \> \geq 1-\ve,\\
  & 0 \leq G \leq \id,\; \Delta(G) = \frac{1}{m} \id \bigg\} \bigg\rfloor.
   \end{aligned}\end{equation*}
\end{theorem}
The result reveals a fundamental relation between different sets of operations in the resource theory of coherence, showing that MIO and DIO have the same power in the task of coherence distillation. This correspondence is in fact surprising: not only is DIO a strict subset of MIO, it is also known that MIO is strictly more powerful than DIO in state transformations \cite{chitambar_2016-1,marvian_2016}, that there exist entropic coherence monotones under DIO which are not monotones under MIO \cite{chitambar_2016-1}, and that the two sets can exhibit different operational capabilities in tasks such as coherence dilution \cite{zhao_2018}. Furthermore, since MIO constitutes the largest class of free operations in the resource theory of coherence, the result is of practical relevance as it shows that using DIO is sufficient to achieve the best rates of distillation achievable under any class of free operations.

We will now show that the quantities introduced in Theorem \ref{thm:distillable_SDP} admit alternative characterizations. In particular, we will express the fidelity of distillation as a measure related to the family $\T{m}$ introduced before, and the one-shot distillable coherence as a quantum hypothesis testing problem. To do so, we will need to optimize over a larger set of matrices than the incoherent states $\I$: namely, the set $\J \coloneqq \lset X \bar \Tr(X) = 1,\;\Delta(X) = X \rset$ of unit-trace diagonal Hermitian matrices, and analogously the set $\J\*\*$ of unnormalized diagonal matrices. We then define the quantities
\begin{equation}\begin{aligned}
  \TJ{m}(\rho) &\coloneqq \min_{X \in \J\*\*}\; m \Tr\left(\rho-X\right)_+ + \Tr\left(\rho-X\right)_-\\
  =&  \max \lset \< \rho, W \> \bar  -\mathbbm{1} \mleq W \mleq m\mathbbm{1},\; \Delta(W) = 0 \rset,
\end{aligned}\end{equation}
in analogy with the measures $\T{m}$. Following the proof of Prop. \ref{prop:Tm_monotonic} one can easily see that $\TJ{m}$ are also faithful strong monotones under MIO. Let us now consider the \textit{hypothesis testing relative entropy} $D^\ve_H$ \cite{wang_2012,tomamichel_2013}, defined as
 \begin{equation}\begin{aligned}
\label{hypothesis testing definition}
D_H^\ve(\rho||\sigma) \coloneqq-\log\min \{  \< M, \sigma\> \;&|\; 0\le M\le \id,\\
&1- \< M, \rho \> \le\ve \}.
\end{aligned}\end{equation}
In the setting of quantum hypothesis testing, one is interested in distinguishing between two quantum states --- $\rho$ and $\sigma$ --- by performing a test measurement $\{M, \id - M\}$ where $0 \mleq M \mleq \id$. The probability of incorrectly accepting state $\sigma$ as true (type-I error) is given by $\<\id - M, \rho\>$, and the probability of incorrectly accepting state $\rho$ as true (type-II error) is given by $\< M, \sigma \>$ \cite{hayashi_2016}. The quantity $D_H^\ve(\rho||\sigma)$ then characterizes the minimum type-II error while constraining the type-I error to be no greater than $\ve$. Alternatively, $D_H^\ve(\rho||\sigma)$ can be viewed as the operator-smoothed version of min-relative entropy \cite{datta_2009,buscemi_2010-1}. Using this quantity, we can obtain the following result.
\begin{proposition}\label{prop:hypothesis_testing}
The fidelity of coherence distillation and one-shot $\ve$-error distillable coherence under $\O \in \{ \MIO, \DIO \}$ admit a characterization as the semidefinite programs
\begin{equation}\begin{aligned}
   F_{\O}(\rho,m) &= \frac{1}{m} \left(\TJ{m-1}(\rho) + 1\right)\,,\\
   C_{d,\rm \O}^{(1),\ve}(\rho) & = \min_{X \in \J} D_H^\ve(\rho||X) - \delta\,,
\end{aligned}\end{equation}
  where $\delta \geq 0$ is the least number such that the solution corresponds to the logarithm of an integer.
\end{proposition}
Although the optimization in the above problems is over matrices which are not necessarily positive semidefinite, one can show that if one of the problems admits a positive semidefinite optimal solution, then so does the other. In the particular case of $m=d$, not only does the fidelity of distillation simplify to an optimization over $\I$, but combining Prop. \ref{prop:hypothesis_testing} with Thm. 1 of Ref. \cite{bu_2017} we know that in fact $F_{\O}(\rho,d)$ is the same for any $\O\in\{\MIO,\DIO,\SIO,\IO\}$. However, the case of interest is when such a property holds for any value of $m$ --- we will now show that this is true for all pure states, significantly simplifying the computation of the above quantities.

We first notice that each $\T{m}$ provides an upper bound on the corresponding $\TJ{m}$, giving $F_{\MIO}(\rho, m) \leq \frac{1}{m} \left(\T{m-1}(\rho) + 1\right)$.
To show that this bound is in fact tight for all pure states, we consider different sets of operations --- SIO as well as IO. Pure state transformations under IO and SIO are known to be fully characterized by majorization relations \cite{nielsen_1999,winter_2016,chitambar_2016-1,zhu_2018-1}, which allow us to lower bound the fidelity of distillation and obtain the following result.
\begin{theorem}\label{prop:all_F_equal}
For any pure state $\ket\psi$, any integer $m \geq 1$, and $\O \in \{ \MIO, \DIO, \SIO, \IO \}$,
\begin{equation}\begin{aligned}
  F_{\O}(\psi,m) &= \frac{1}{m} \mnorm{\ket\psi}^2.
\end{aligned}\end{equation}
\end{theorem}
This extends the operational equivalence between MIO and DIO in coherence distillation to the strictly smaller set SIO, and has several important consequences.
Firstly, it shows that the one-shot distillable coherence of pure states under any of the classes of operations $\O \in \{\MIO, \DIO, \SIO, \IO\}$ is exactly the same, and in fact can be expressed as the quantum hypothesis testing problem $C_{d,\rm \O}^{(1),\ve}(\psi) = \min_{\sigma \in \I} D_H^\ve(\psi||\sigma) - \delta$ with $\delta$ as before. Secondly, we can use the properties of the $m$-distillation norm to obtain exact formulas for the one-shot distillable coherence. In particular, noting that $\mnorm{\ket\psi} = \sqrt{m}$ (or equivalently $F_\O(\psi,m) = 1$) if and only if $\lnorm{\ket\psi}{\infty} \leq \frac{1}{\sqrt{m}}$, we see that the zero-error distillable coherence is given by
\begin{equation}\begin{aligned}
  C^{(1),0}_{d,\O} (\psi)& = \log \left\lfloor \lnorm{\ket\psi}{\infty}^{-2} \right\rfloor.
\end{aligned}\end{equation}

Relating the $m$-distillation norm with the fidelity of distillation also allows us to more easily make quantitative statements about the distillability of pure states on average. For example, in Ref. \cite{johnston_2017-1} it was shown that the proportion of pure states with respect to the Haar measure for which $\norm{\ket\psi}{[2]} = 1$ is given by $1 - d \,2^{1-d}$, which sharply tends to $1$ as $d$ increases. In light of our results, this then shows that, with growing dimension, only an exponentially small fraction of pure states are one-shot undistillable, while a significant majority of pure states satisfy $C_{d,\rm \O}^{(1),0}(\psi) \geq 1$ and therefore allow for a zero-error one-shot distillation of at least one bit of coherence.

The results of our work have important consequences beyond the one-shot regime, in particular for the {\it asymptotic} reversibility of state transformations in the resource theory of coherence --- that is, the question whether the amount of coherence which can be distilled from a number of copies of a state $\rho$ (distillable coherence $C^\infty_d$) is the same as the amount of coherence needed to prepare the same number of copies (coherence cost $C^\infty_c$) in the asymptotic limit of an arbitrarily large number of i.i.d. copies. It is known that the resource theory of coherence is reversible under MIO \cite{winter_2016,zhao_2018}, but irreversible under IO as we have $C^\infty_{d,\IO} (\rho) < C^\infty_{c,\IO}(\rho)$ in general \cite{winter_2016}. Recently, it has been claimed that $C^\infty_{c,\DIO} (\rho) = C^\infty_{c,\MIO}(\rho) = C_r(\rho)$ \cite{zhao_2018}, although a complete proof of this fact did not appear until \cite{chitambar_2017-1}. Our result in Thm. \ref{thm:distillable_SDP} in particular shows that $C_{d,\DIO}^{(1),\ve}(\rho) = C_{d,\MIO}^{(1),\ve}(\rho)$ and therefore $C^\infty_{d,\DIO} (\rho) = C^\infty_{d,\MIO}(\rho) = C_r(\rho)$, complementing the claims of Ref. \cite{zhao_2018} and strengthening the asymptotic results of Ref. \cite{chitambar_2017-1} by showing their applicability even in the one-shot case. The fact that state transformations are indeed reversible under DIO and the maximal set of operations MIO is not necessary for full reversibility contrasts with other resource theories such as entanglement, where the only set of operations known to provide asymptotic reversibility is strictly larger than the maximal set \cite{vidal_2001,vidal_2002-2,brandao_2008-1,wang_2017-1}.


\section{Conclusions}
 We have characterized the operational task of one-shot coherence distillation for several classes of free operations, showing in particular that MIO and DIO have the same power in this task, and providing computable expressions for the rates of distillation in terms of a quantum hypothesis testing problem. Further, we have introduced a family of coherence measures and related it to the achievable fidelity of distillation. By quantifying the introduced measures exactly on pure states and showing that they reduce to a class of much simpler vector norms, we have obtained a full characterization of one-shot coherence distillation from pure states, and established that in this case all relevant sets of operations are equally useful.

 Our work unveils several new features of the resource theory of coherence and contributes to a better understanding of the properties of the different sets of free operations, as well as generalizes and provides an operational interpretation to several coherence monotones. This yields further insight on how quantum coherence can be created and transformed in the realistic setting of finitely many quantum states available.

 The possible applications of one-shot coherence distillation are multifold. Notably, the framework presented herein can be used to precisely characterize the experimentally feasible rates at which maximally coherent states, often employed as ``currency'' in operational tasks, can be prepared. One such application is randomness extraction \cite{yuan_2015}, which can be implemented by distilling the coherence of a quantum state followed by a measurement generating uniformly random bits. Another promising way of utilizing one-shot coherence distillation is to enable coherent state preparation for direct use in quantum key distribution and quantum algorithms \cite{scarani_2009,coles_2016,hillery_2016}.
 Furthermore, the comparison of the operational capabilities of different classes of operations provides, in particular, new insight about the relatively unexplored class DIO, whose relation with the so-called thermal operations could find use in the resource theory of quantum thermodynamics \cite{lostaglio_2015-1,cwiklinski_2015}.\\[\baselineskip]

\begin{acknowledgments}

\textit{Note.} --- During the completion of this work, we became aware of an independent work by E. Chitambar \cite{chitambar_2017-1} where the author considers the asymptotic properties of state transformations under DIO and in particular obtains a different proof that the asymptotic rate of coherence distillation under MIO and DIO is the same.

We are grateful to Eric Chitambar, Min-Hsiu Hsieh, Nathaniel Johnston, Ludovico Lami, Andreas Winter, and Wei Xie for discussions. BR and GA acknowledge financial support from the European Research Council (ERC) under the Starting Grant GQCOP (Grant No.~637352). KF and XW were partly supported by the Australian Research Council (Grant No. DP120103776 and No. FT120100449).

\end{acknowledgments}

\sloppy
\bibliographystyle{apsrev4-1}
\bibliography{main}
\fussy



\onecolumngrid
\begin{center}
\vspace*{\baselineskip}
{\textbf{\large Supplemental Material: \\[3pt] One-shot coherence distillation}}\\[1pt] \quad \\
\end{center}

\renewcommand{\theequation}{S\arabic{equation}}
\setcounter{equation}{0}
\setcounter{figure}{0}
\setcounter{table}{0}
\setcounter{section}{0}
\setcounter{page}{1}
\makeatletter

\section{Properties of the family of coherence monotones}

\begingroup
\renewcommand\theproposition{1}
\begin{proposition}
For each $m\geq 1$, $\T{m}$ is a faithful and convex measure of coherence satisfying strong monotonicity under MIO.
 \end{proposition}
 \endgroup
 \begin{remark}
By strong monotonicity under MIO, we understand the notion that
\begin{equation}\begin{aligned}
  \sum_i \Tr\left[\Theta_i(\rho)\right]\; \T{m}\left(\frac{\Theta_i(\rho)}{\Tr\left[\Theta_i(\rho)\right]}\right) \leq \T{m}(\rho)
\end{aligned}\end{equation}
where $\Theta_i$ are completely positive, trace non-decreasing maximally incoherent operations such that $\sum_i\Theta_i$ is trace preserving. This notion of strong monotonicity has been considered e.g. in \cite{piani_2016} and is stronger than the one commonly encountered in the literature, wherein each $\Theta_i$ is identified with the action of a single Kraus operator \cite{baumgratz_2014}.
\end{remark}
 \begin{proof}
Let
$\Gamma_\C(\rho) = \inf \lset \lambda \geq 0 \bar \rho \in \lambda \conv(\C) \rset = \sup \lset \< \rho, W \> \bar \< W, \sigma\> \leq 1 \;\forall \sigma \in \C \rset$
denote the gauge function corresponding to the set $\conv(\C)$. Each such function is convex. From the dual characterization of $\T{m}$ in Eq. \eqref{eq:family_Tm} it follows that we have $\T{m}(\rho) = \Gamma_{\C_m}(\rho)$ where $\C_m = (-\DD) \cup \frac{1}{m} \DD \cup \I\*\*$. Since $\Gamma_{\C_m}$ is zero only for $X \in \I\*\*$, each $\T{m}$ is faithful. Since $\rho \in \conv(\C_m) \Rightarrow \frac{1}{\Tr\Lambda(\rho)}\Lambda(\rho) \in \conv(\C_m)$ for each $\Lambda \in \text{MIO}$, each $\T{m}$ is strongly monotonic under MIO \cite[Thm. 20]{regula_2018}.
 \end{proof}


\begingroup
\renewcommand\thetheorem{1}
\begin{theorem}
For any pure state $\ket\psi$ and any $m \geq 1$ we have
\begin{equation}\begin{aligned}
  \T{m-1}(\psi) = \mnorm{\ket\psi}^2 - 1.
\end{aligned}\end{equation}
\end{theorem}
To prove this Theorem we will use the following result, which appears as Thm. 10 in \cite{regula_2018}.
\begin{lemma}\label{lemma:robustness}
Let $\Q \coloneqq \conv \lset \proj{\phi} \bar \ket{\phi} \in \V \rset$ where $\V \subseteq \CC^d$ is a compact set such that $\ket{\phi} \in \V \Rightarrow e^{i\theta}\ket\phi \in \V \; \forall \theta \in \RR$ and $\operatorname{span}(\V) = \CC^d$. For any pure state $\ket\psi$, it then holds that
\begin{equation}\begin{aligned}
  \max \lset \braket{\psi | W | \psi} \bar W \mgeq 0, \;\< W, X \> \leq 1 \; \forall X \in \Q \rset = \max \lset \cbraket{\psi | w}^2 \bar \cbraket{w | v} \leq 1 \; \forall \ket{v} \in \V \rset.
\end{aligned}\end{equation}
\end{lemma}
\begin{proof}[Proof of Theorem 1]
Consider the set $\Qm \coloneq \I \cup \frac{1}{m}\DD$, and define the robustness with respect to $\Qm$ as 
\begin{equation}\begin{aligned}
  R_\Qm(\rho) + 1 \coloneqq& \max \lset \< \rho, W \> \bar W \mgeq 0,\; \<W, Q\> \leq 1 \, \forall\, Q \in \Qm \rset \\
\end{aligned}\end{equation}
where the term $+1$ is added to keep the notation consistent with quantities such as the robustness of coherence ($R_\I$). Notice that with the change of variables $W' = W - \id$ we simply have
\begin{equation}\begin{aligned}
  R_\Qm(\rho) &= \max \lset \< \rho, W \> \bar 0 \mleq W \mleq m \id,\; \Delta(W) \mleq \id \rset - 1\\
  &= \max \lset \< \rho, W' \> \bar -\id \mleq W' \mleq (m-1) \id,\; \Delta(W') \mleq 0\rset\\
  &= \T{m-1} (\rho).
\end{aligned}\end{equation}

Let us then define the sets
\begin{equation}\begin{aligned}
  \B &\coloneq \lset \ket\psi \bar \proj\psi \in \I \rset\\
  \N_m &\coloneq \lset \ket\psi \bar \lnorm{\ket\psi}{2} = \frac{1}{\sqrt{m}} \rset\\
  \V_m &\coloneq \lset \ket \psi \bar \ket\psi \in \B \cup \N_m \rset\\
\end{aligned}\end{equation}
using which we get
\begin{equation}\begin{aligned}
  \Qm &= \lset \proj\psi \bar \ket\psi \in \V_m \rset.
\end{aligned}\end{equation}
What Lemma \ref{lemma:robustness} says is that the robustness measure $R_\Qm$ corresponding to a set of this form reduces on pure states to the gauge function of the set $\conv(\V_m)$; specifically, $R_\Qm(\psi) + 1 = \Gamma_{\V_m}(\ket\psi)^2$, with $\Gamma_{\V_m}(\ket\psi) = \inf \lset \lambda \geq 0 \bar \ket\psi \in \lambda \conv(\V_m) \rset$. By standard results in convex analysis (see e.g. \cite{rockafellar_1970}, 16.4.1 and 15.1.2), the gauge function of the set $\conv(\V_m)$ can be given as
\begin{equation}\begin{aligned}
  \Gamma_{\V_m} (\ket\psi) &= \inf_{\ket\psi = \ket{x} + \ket{y}} \Gamma_\B (\ket{x}) + \Gamma_{\N_m} (\ket{y})\\
  &=\inf_{\ket\psi = \ket{x} + \ket{y}} \lnorm{\ket{x}}{1} + \sqrt{m}\lnorm{\ket{y}}{2}\\
  &= \mnorm{\ket\psi}.
\end{aligned}\end{equation}
The result then follows by Lemma \ref{lemma:robustness}.
\end{proof}
\endgroup


\begingroup
\renewcommand\thetheorem{2}
\begin{theorem}
For a pure state $\ket\psi$, let $\psi^\downarrow_{1:k}$ denote the vector consisting of the $k$ largest (by magnitude) coefficients of $\ket\psi$, and analogously let $\psi^\downarrow_{k+1:d}$ denote the vector of the $d-k$ smallest coefficients of $\ket\psi$, with $\psi^\downarrow_{1:0}$ being the zero vector. Then, for any pure state $\ket\psi$ and any integer $m \in \{1,\ldots,d\}$ we have
\begin{equation*}\begin{aligned}
  \mnorm{\ket\psi} = \lnorm{\psi^\downarrow_{1:m-k^\star}}{1} + \sqrt{k^\star} \lnorm{\psi^\downarrow_{m-k^\star+1:d}}{2},
\end{aligned}\end{equation*}
where $k^\star = \displaystyle\argmin_{1 \leq k \leq m} \mbox{$\frac{\lnorm{\psi^\downarrow_{m-k+1:d}}{2}^2}{k}$}.$
\end{theorem}
\begin{proof}
Without loss of generality, assume that $\ket\psi$ has non-negative coefficients arranged in non-increasing order. Let
\begin{equation}\begin{aligned}
  L_1 &\coloneq \lnorm{\psi^\downarrow_{1:m-k^\star}}{1}\\
  L_2 &\coloneq \lnorm{\psi^\downarrow_{m-k^\star+1:d}}{2}
\end{aligned}\end{equation}
and take the feasible solutions
\begin{equation}\begin{aligned}
  \ket{x} &= \left( \psi_1 - \frac{L_2}{\sqrt{k^\star}}, \ldots, \psi_{m-k^\star} - \frac{L_2}{\sqrt{k^\star}},\, 0, \ldots, 0\right)^T\\
  \ket{y} &= \Bigg(\underbrace{\frac{L_2}{\sqrt{k^\star}},\, \ldots,\,\frac{L_2}{\sqrt{k^\star}}}_{m-k^\star \text{ times}}, \psi_{m-k^\star+1},\ldots,\psi_d \Bigg)^T.
\end{aligned}\end{equation}
Since $\ket\psi = \ket{x}+\ket{y}$, we have
\begin{equation}\begin{aligned}
  \mnorm{\ket\psi} &\leq \lnorm{\ket{x}}{1} + \sqrt{m} \lnorm{\ket{y}}{2}\\
  &= L_1 - (m-k^\star) \frac{L_2}{\sqrt{k^\star}} + \sqrt{m} \sqrt{ (m-k^\star) \frac{L^2_2}{k^\star} + L^2_2}\\
  &= L_1 - (m-k^\star) \frac{L_2}{\sqrt{k^\star}} + \sqrt{ \frac{m(m-k^\star)}{k^\star} L^2_2 + m L^2_2}\\
  &= L_1 - (m-k^\star) \frac{L_2}{\sqrt{k^\star}} + \frac{m}{\sqrt{k^\star}} L_2\\
  &= L_1 + \sqrt{k^\star} L_2.\\
\end{aligned}\end{equation}

Now, by strong Lagrange duality we have
\begin{equation}\begin{aligned}\label{eq:norm_dual}
  \mnorm{\ket\psi} 
  &= \max \lset \cbraket{\psi|x} \bar \lnorm{\ket{x}}{\infty}\leq 1,\;\lnorm{\ket{x}}{2} \leq \sqrt{m} \rset.
\end{aligned}\end{equation}
Choosing
\begin{equation}\begin{aligned}
  \ket{x} = \Bigg(\underbrace{\vphantom{\Big(}1, \ldots, 1}_{m-k^\star \text{ times}}, \frac{\sqrt{k^\star}}{L_2}\psi_{m-k^\star+1},\ldots,\frac{\sqrt{k^\star}}{L_2}\psi_d \Bigg)^T
\end{aligned}\end{equation}
we have that
\begin{equation}\begin{aligned}
  \lnorm{\ket{x}}{2} &= \sqrt{m-k^\star + \frac{k^\star}{L^2_2} L^2_2} = \sqrt{m}\\
  \lnorm{\ket{x}}{\infty} &= \max \left\{ 1,\; \frac{\sqrt{k^\star}}{L_2} \psi_{m-k^\star+1} \right\}.
\end{aligned}\end{equation}
To see that $\lnorm{\ket{x}}{\infty}$ is upper bounded by $1$, consider first the case when $k^\star = 1$. We then have
\begin{equation}\begin{aligned}
  \frac{\sqrt{k^\star}}{L_2} \psi_{m-k^\star+1} = \frac{\psi_{m}}{\psi_m + \lnorm{\psi^\downarrow_{m+1:d}}{2}^2} \leq 1.
\end{aligned}\end{equation}
Assume now that $k^\star > 1$ and consider the following chain of equivalent inequalities:
\begin{equation}\begin{aligned}\label{eq:chain_ineq}
  \frac{\lnorm{\psi^\downarrow_{m-k^\star+1:d}}{2}^2}{k^\star} &\leq \frac{\lnorm{\psi^\downarrow_{m-k^\star+2:d}}{2}^2}{k^\star-1}\\
    \psi^2_{m-k^\star+1} + \lnorm{\psi^\downarrow_{m-k^\star+2:d}}{2}^2 &\leq \frac{k^\star}{k^\star-1} \lnorm{\psi^\downarrow_{m-k^\star+2:d}}{2}\\
    \psi^2_{m-k^\star+1} &\leq \frac{1}{k^\star-1} \left(1 - \lnorm{\psi^\downarrow_{1:m-k^\star+1}}{2}^2\right)\\
    k^\star \psi^2_{m-k^\star+1} &\leq \left(1 - \lnorm{\psi^\downarrow_{1:m-k^\star}}{2}^2\right)\\
    \sqrt{k^\star} \psi_{m-k^\star+1} &\leq L_2
\end{aligned}\end{equation}
where the first line follows by hypothesis, and in the third and fifth lines we have used the fact that $\ket\psi$ is a normalized pure state. This means that
\begin{equation}\begin{aligned}
  \frac{\sqrt{k^\star}}{L_2} \psi_{m-k^\star+1} \leq \frac{L_2}{L_2} = 1
\end{aligned}\end{equation}
and so $\ket{x}$ is indeed a feasible solution to the dual optimization problem. We then have
\begin{equation}\begin{aligned}
  \mnorm{\ket\psi} &\geq \cbraket{\psi|x}\\
  &= L_1 + \sqrt{k^\star} L_2
\end{aligned}\end{equation}
which completes the proof.
\end{proof}
\endgroup


\section{Coherence distillation}

\begingroup
\renewcommand\thetheorem{3}
\begin{theorem}
  If $\O \in \{ \MIO, \DIO \}$, then for any state $\rho \in \DD$, the fidelity of coherence distillation and one-shot $\ve$-error distillable coherence can both be written as the following semidefinite programs:
  \begin{equation}\begin{aligned}
    F_{\O}(\rho,m)  =  \max \left\{ \<G, \rho\> \;\left|\; 0 \leq G \leq \id,\ \Delta(G) = \frac1m \id \right.\right\},
     \end{aligned}\end{equation}\vspace{-\baselineskip}\begin{equation*}\begin{aligned}\label{eq:dist_coherence}
  C_{d,\rm \O}^{(1),\ve}(\rho)  = \log \bigg\lfloor \max \bigg\{& m \;\bigg|\; \< G, \rho \> \geq 1-\ve,\; 0 \leq G \leq \id,\; \Delta(G) = \frac{1}{m} \id \bigg\} \bigg\rfloor.
   \end{aligned}\end{equation*}
\end{theorem}
\endgroup


\begin{proof} Consider first the class MIO. Denote $J$ as the Choi matrix of operation $\Pi$. Then
  \begin{align}
  \tr \Pi(\rho)\Psi_m = \tr (\tr_A J \cdot \rho^{T}\ox \id_m) \Psi_m = \tr \left[\tr_B(J \cdot \id_d\ox \Psi_m)\right] \rho^{T}.
  \end{align}
  Denote $U_\pi$ as the unitary that permutes the basis $\{\ket{i}\}$ to $\{\ket{\pi(i)}\}$. Then $\Psi_m$ is invariant under all $U_\pi$. Exploiting such symmetry, w.l.o.g., we can take $J = \sum \id_d\ox U_\pi J \id_d \ox U_\pi^\dagger$. Thus $J = R \ox \id_m +  Q \ox (\Psi_m - \frac1m \id_m)$, since $a\id_m + b (\Psi_m - \frac1m \id_m)$ is the only form of operator invariant under all $U_\pi$. Since $\tr_B J = \id_d$, we have $R = \frac1m \id_d$. To constrain $\Pi$ to be in MIO, we require that it maps any incoherent state to an incoherent state, or equivalently $\Pi(\proj{i}) \in \I$, which is only possible if $\Delta(Q) = 0$. Thus the optimal operation has the structure
  \begin{align}\label{Optimal J}
    J = Q \ox (\Psi_m - \frac1m \id_m) + \id_d \ox \frac1m \id_m,\quad \Delta(Q) = 0.
  \end{align}
  Then, $J \geq 0$ if and only if $-\frac{1}{m-1}\id_d \leq Q \leq \id_d$. By direct calculation, $\tr \Pi(\rho) \Psi_m = (1-\frac1m)\tr Q\rho^{T} + \frac1m$. Then we have
  \begin{align}
     F_{\rm MIO}(\rho,m) = \max \left\{\left(1-\frac1m\right)\tr Q\rho + \frac1m \;\left|\; -\frac{1}{m-1}\id_d \leq Q \leq \id_d, \Delta(Q) = 0 \right.\right\}.
   \end{align}
   Replacing $G = \frac{m-1}{m}Q + \frac1m \id_d$, we have
     \begin{align}
     F_{\rm MIO}(\rho,m) = \max \left\{\tr G \rho \;\left|\; 0 \leq G \leq \id_d, \Delta(G) = \frac1m \id_d\right.\right\}.
   \end{align}
   Following the definition of one-shot distillable coherence, we have the SDP for $C_{d,\rm MIO}^{(1),\ve}(\rho)$.

   The condition that $\Pi$ is in DIO corresponds to $\Pi \in \MIO$ and $\Delta(\Pi(\ket{i}\!\bra{j})) = 0 \; \forall i \neq j$. From Eq.~\eqref{Optimal J}, we can see that the optimal operation is in fact already in DIO. Thus we have
  \begin{align}
    F_{\rm MIO}(\rho,m) = F_{\rm DIO}(\rho,m), \quad C_{d,\rm MIO}^{(1),\ve}(\rho) = C_{d,\rm DIO}^{(1),\ve}(\rho).
  \end{align}

     Note that the optimal operation is given by
   \begin{align}
   \Pi(\rho) = (\tr G \rho) \Psi_m + (1-\tr G\rho) \frac{\id_m - \Psi_m}{m-1},\ 0 \leq G \leq \id_d,\ \Delta(G) = \frac1m \id_d . \end{align}
\end{proof}

\begingroup
\renewcommand\thetheorem{4}
\begin{theorem}
For any pure state $\ket\psi$, any integer $m \geq 1$, and $\O \in \{ \MIO, \DIO, \SIO, \IO \}$,
\begin{equation}\begin{aligned}
  F_{\O}(\psi,m) &= \frac{1}{m} \mnorm{\ket\psi}^2.
\end{aligned}\end{equation}
\end{theorem}
\endgroup
\begin{proof}
Using the fact that $\SIO \subset \DIO \subset \MIO$, the inclusion $\I \subset \J$, as well as Theorem 1, we have
\begin{equation}\begin{aligned}
  F_{\SIO}(\psi,m) \leq F_{\DIO}(\psi,m) \leq F_{\MIO}(\psi,m) \leq \frac{1}{m} (\T{m-1}(\psi)+1) = \frac{1}{m}\mnorm{\ket\psi}^2.
\end{aligned}\end{equation}

We will now show that $F_\SIO(\psi,m) \geq \frac{1}{m} \mnorm{\ket{\psi}}^2$, thus concluding the proof. To this end, begin by noting that, by the concavity of the fidelity, to characterize the optimal fidelity of distillation of a pure state it is sufficient to consider transformations into pure states \cite{vidal_2000-1}. It is known that the deterministic transformation from $\ket\psi$ to another pure state $\ket\eta$ is possible by SIO or IO if and only if the vector of squared coefficients of $\ket\psi$ is majorized by the vector of squared coefficients of $\ket\eta$ \cite{nielsen_1999,winter_2016,chitambar_2016,zhu_2018-1}, that is,
\begin{equation}\begin{aligned}
  \sum_{i=1}^{k} \psi^2_i \leq \sum_{i=1}^{k} \eta^2_i \quad \forall k \in \{1, \ldots, d\}
\end{aligned}\end{equation}
where we have assumed without loss of generality that the coefficients of $\ket\psi$ and $\ket\eta$ are non-negative and arranged in non-increasing order.

Assuming that $m \leq d$ (since otherwise we have $\mnorm{\ket\psi} = \norm{\ket\psi}{[d]}$, so we can take $m=d$ and proceed analogously), let $\begin{displaystyle}k^\star = \argmin_{1 \leq k \leq m} \frac{1}{k} \lnorm{\psi^\downarrow_{m-k+1:d}}{2}^2\end{displaystyle}$ and $L_2 = \lnorm{\psi^\downarrow_{m-k^\star+1:d}}{2}$. Consider then the $m$-dimensional state
\begin{equation}\begin{aligned}
  \ket{\eta} = \Bigg( \psi_1,\, \ldots,\, \psi_{m-k^\star},\, \underbrace{\vphantom{\bigg(}\frac{L_2}{\sqrt{k^\star}}, \ldots, \frac{L_2}{\sqrt{k^\star}}}_{k^\star \text{ times}}\Bigg)^T.
\end{aligned}\end{equation}
To show that $\ket\psi$ can be transformed into $\ket\eta$ by IO/SIO, recall from \eqref{eq:chain_ineq} that
\begin{equation}\begin{aligned}
    \psi_{m-k^\star+1}^2 &\leq \frac{L_2^2}{k^\star},
\end{aligned}\end{equation}
which means that the squared coefficients of $\ket\eta$ majorize the squared coefficients of $\ket\psi$. This then gives
\begin{equation}\begin{aligned}
  F_\SIO(\psi,m) &\geq \cbraket{\Psi_m|\eta}^2\\
  &= \frac1m \left(\lnorm{\psi^\downarrow_{1:m-k^\star}}{1} + \sqrt{k^\star} \lnorm{\psi^\downarrow_{m-k^\star+1:d}}{2}\right)^2\\
  &= \frac{1}{m} \mnorm{\ket{\psi}}^2
\end{aligned}\end{equation}
by Theorem 2, as required.
\end{proof}

\begin{remark}
The fact that $F_\SIO (\psi,m) = \frac{1}{m} \mnorm{\ket{\psi}}^2$ can be equivalently obtained by using the result of Vidal et al. \cite{vidal_2000-1}, where the fidelity of pure-state transformations under LOCC is considered.
\end{remark}


\begingroup
\renewcommand\theproposition{2}
\begin{proposition}
For any state $\rho$, the one-shot distillable coherence under MIO and DIO corresponds to
\begin{equation}\begin{aligned}
  C_{d,\rm \O}^{(1),\ve}(\rho) = \min_{M \in \J} D^{\ve}_H(\rho\|M) - \delta
\end{aligned}\end{equation}
where $\delta \geq 0$ is the least number such that the solution corresponds to the logarithm of an integer, and $\J = \lset X \bar \Tr(X) = 1,\; X = \Delta(X) \rset$.
\end{proposition}
\begin{proof}
  To begin with, note that
  \begin{equation}\begin{aligned}
    \min \lset k \bar \Delta(W) = k \id \rset = \max \lset \<M, \Delta(W) \> \bar \Tr(M) = 1 \rset
  \end{aligned}\end{equation}
  by strong Lagrange duality. From this, we can obtain
  \begin{equation}\begin{aligned}
    C_{d,\rm \O}^{(1),\ve}(\rho)  &= - \log \min_{\substack{\< \rho, G \> \geq 1 - \ve\\0 \mleq G \mleq \id}} \min \lset k \bar \Delta(G) = k \id \rset - \delta\\
 &= - \log \min_{\substack{\< \rho, G \> \geq 1 - \ve\\0 \mleq G \mleq \id}} \,\max_{\substack{\Tr(M)=1}} \< M, \Delta(G) \> - \delta\\
  &= - \log \max_{\substack{\Tr(M)=1}} \min_{\substack{\< \rho, G \> \geq 1 - \ve\\0 \mleq G \mleq \id}} \< \Delta(M), G \> - \delta\\
 &= \min_{\substack{\Tr(M)=1\\M=\Delta(M)}} - \log \min_{\substack{\< \rho, G \> \geq 1 - \ve\\0 \mleq G \mleq \id}} \< M, G \> - \delta\\
 &= \min_{M \in \J} D^{\ve}_H(\rho\|M)- \delta
  \end{aligned}\end{equation}
  where the third equality follows by Sion's minimax theorem \cite{sion_1958} and the self-duality of the diagonal map $\Delta$. Here, without loss of generality we take $\log x = -\infty$ for any $x \leq 0$.
\end{proof}
\begin{remark}One can show in the same way that pure states satisfy
\begin{equation}\begin{aligned}
  C_{d,\rm \O}^{(1),\ve}(\psi) &= \min_{\sigma \in \I} D_H^\ve(\psi||\sigma) - \delta\\
  &= -\log \min_{\substack{\< \psi, G \> \geq 1 - \ve\\0 \mleq G \mleq \id}}\, \norm{\Delta(G)}{\infty} - \delta
\end{aligned}\end{equation}
  since we know from Thm. 4 that $F_\O(\psi,m) = \frac{1}{m} \left(\T{m}(\psi) + 1\right)$.
  \end{remark}
\endgroup


\end{document}